\documentclass[a4paper,USEnglish]{llncs}

\usepackage[T1]{fontenc}
\usepackage[utf8]{inputenc}
\usepackage{lmodern}

\usepackage{graphicx,algorithm,algpseudocode,algorithmicx}
\usepackage{amsmath,amssymb,mathtools}
\usepackage{xspace}
\usepackage{colortbl}
\usepackage{xcolor}
\usepackage{multicol}
\usepackage{pgfplots}
\usepackage{setspace}
\usepackage{hyperref}
\usepackage{cleveref}

\let\leftold\left
\let\rightold\right
\renewcommand{\left}{\mathopen{}\mathclose\bgroup\leftold}
\renewcommand{\right}{\aftergroup\egroup\rightold}

\crefname{algorithm}{Algorithm}{Algorithms}
\crefname{definition}{Definition}{Definitions}
\crefname{figure}{Figure}{Figures}
\crefname{observation}{Observation}{Observations}
\crefname{section}{Section}{Sections}
\crefname{theorem}{Theorem}{Theorems}
\crefname{lemma}{Lemma}{Lemmat}
\crefname{corollary}{Corollary}{Corollary}
\crefname{equation}{}{}
\crefname{enumi}{}{}

\newcommand\algo[1]{\ensuremath{\textsc{Alg}\left[{#1}\right]}\xspace}
\newcommand\algoinline[1]{\ensuremath{\textsc{Alg}[{#1}]}\xspace}

\newcommand\ie{i.\,e.\xspace}

\def\alg{\ensuremath{\textsc{Alg}}\xspace}
\def\opt{\ensuremath{\textsc{Opt}}\xspace}
\def\OS{{\sc Online Search}\xspace}

\def\O2S{{\sc Online $2$-Search}\xspace}

\def\OWT{{\sc One-Way Trading}\xspace}

\title{Advice Complexity of the Online Search Problem}
\author{Jhoirene Clemente\inst{1}\thanks{Supported by ERDT Scholarship. Sandwich program funded by PCIEERD-BCDA.} \and Juraj Hromkovi\v{c}\inst{2}\thanks{Supported by SNF grant 200021-146372.} \and Dennis Komm\inst{2} \and Christian Kudahl\inst{3}\thanks{Supported by the Villum Foundation and the Stibo-Foundation.}}
 \institute{Department of Computer Science,\\University of the Philippines Diliman, Philippines,\\\email{jbclemente@up.edu.ph} \and
 Department of Computer Science,\\ETH Z\"urich, Switzerland,\\\email{\{juraj.hromkovic,dennis.komm\}@inf.ethz.ch} \and
 Department of Mathematics and Computer Science,\\University of Southern Denmark, Denmark,\\\email{kudahl@imada.sdu.dk}}
 
\begin{document}
 
\maketitle
 
\begin{abstract}
  The online search problem is a fundamental problem in finance.  The numerous
  direct applications include searching for optimal prices for commodity trading 
  and trading foreign currencies.  In this paper, we analyze the advice
  complexity of this problem.
  In particular, we are interested in identifying the minimum amount of information needed in order to 
  achieve a certain competitive ratio.
  We design an algorithm that reads $b$ bits of advice and achieves a
  competitive ratio of $(M/m)^{1/(2^b+1)}$ where $M$ and $m$ are the maximum and
  minimum price in the input.  We also give a matching lower bound.
  Furthermore, we compare the power of advice and randomization for this problem.
\end{abstract}

\section{Introduction}

We study the online search problem (abbreviated \OS), which is formulated as
an online (profit) maximization problem.
For such problems, the input arrives gradually in consecutive time steps.
Each piece of input is called a \emph{request}.  After a request is given, an online
algorithm (also called the \emph{online player}) has to produce a definite
piece of the output, called an \emph{answer}.  Each answer is thus computed
without any knowledge about further requests \cite{BE98}.  The goal is to produce
an output with a \emph{profit} that is as large as possible.
In \OS, the online player searches for the maximum price of a
certain asset that unfolds sequentially.  Suppose the player, in this context a
trader,  would like to transfer its assets from, say, USD to CHF in one transaction.
Each day (formally, each time step), the trader receives a quotation of the current exchange rate and
decides whether to trade on the same day or to wait.  The trading duration is
finite, and it may be known or unknown to the trader.  Formally, we define \OS
as follows.

\begin{definition}[Online Search Problem] \label{def:online_search}
  Let $\sigma = (p_1,p_2,\ldots,p_n)$, with $0<m \leq p_i \leq M$ for all $1
  \leq i \leq n$, be a \emph{sequence of prices} that arrives in an online fashion.
  Here, $M$ and $m$ are upper and lower bounds on the prices, respectively.  For
  each day $i$, price $p_i$ is revealed, and the online player has to choose whether
  to trade on the same day or to wait for the new price quotation on the next
  day.  If the player trades on day $i$, its profit is $p_i$.
  If the player did not trade for the first $n-1$ days, it must accept $p_n$.  The
  player's goal is to maximize the obtained price (\ie, its \emph{profit}).
\end{definition} 

We assume that the parameters $m$ and $M$ for the
price range are fixed and known to the online algorithm in advance.  The duration of the
trading period $n$ is finite and may or may not be known to the online algorithm.
We do not take into account sampling costs in the profit, \ie, the price for
each day is freely given by the market to the trader.  However, some direct
applications of \OS may do require to consider the sampling
costs.  For instance, obtaining prices of a certain product may induce some
cost, either in the form of time or money, from the player.
For a study of such more involved cost variants, we refer the reader to Xu et al.\
\cite{Xu2011}, where the authors considered the accumulated sampling cost while
maximizing the player's profit.

\subsection{Competitive Analysis and Advice Complexity}

\textit{Competitive analysis} was introduced by Sleator and Tarjan
in 1985 \cite{ST85} to analyze the solution quality of online algorithms.  The
measure used in the analysis is called the \textit{competitive ratio}, which
can be  obtained by comparing the profit of the online algorithm to the one of an
optimal offline solution.  The term ``offline'' is used when the whole input
sequence is known in advance.  Note that it is generally not possible
for an online algorithm to compute the optimal
offline solution in advance, because parts of the output have to be specified before the whole
input is known.  It is merely taken into account to analyze the profit that
can hypothetically be obtained if the whole input is known in advance.  The
competitive ratio of an online algorithm is formally defined as follows.

\begin{definition}[Competitive Ratio]\label{def:competitive_analysis}
  Let $\Pi$ be an online maximization problem, let $\alg$ be an online algorithm
  for $\Pi$, and let $c>1$.  $\alg$ is said to be $c$-\textit{competitive} if, for every instance $I$
  of $\Pi$, we have
  \[ c\cdot\mathrm{profit}(\alg(I)) \geq \mathrm{profit}(\opt(I))\;, \]
  where $\mathrm{profit}(\alg(I))$ is the profit of $\alg$ on input $I$,
  and $\mathrm{profit}(\opt(I))$ denotes the optimal offline profit.
\end{definition}

In this paper, we study the \emph{advice complexity} of \OS.  More specifically, we
ask about the additional information both sufficient and necessary in order
to improve the obtainable competitive ratio.  In a way, this approach can
be seen as measuring the \emph{information content} of the problem at hand \cite{HKK10}.

This tool, which was introduced by Dobrev et al.\ in 2008 \cite{DKP08} and then revised
by B\"ockenhauer et al.\ \cite{BKKKM09}, Hromkovi\v{c}
et al.\ \cite{HKK10}, and Emek et al.\ \cite{EFKR11}, is a complementary
tool to analyze online problems.
In order to study the information that is needed in order to outperform purely
deterministic (or randomized) online algorithms, we introduce a trusted source, referred to as an \emph{oracle}, which sees the whole
input in advance and may write binary information on a so-called \emph{advice tape}. 
These \emph{advice bits} are allowed to be any function of the entire input.
The algorithm, which is called an \emph{online algorithm with advice} in this setting,
may then use the advice to compute the output for the given input.
The approach is quantitative
and problem-independent.  In other words, the information supplied can be arbitrary
(as long as it is computable).  This is in particular interesting to give lower bounds for
many other measurements or relaxations of online problems.
More specifically, hardness results in advice complexity give useful negative results
about various semi-online approaches.  If it is for example shown that $O(\log_2 n)$
bits of advice do not help any online algorithm to achieve a better competitive ratio,
this gives a negative answer to questions of the form: Would it help the algorithm to know the
length of the input?  Would it help the algorithm to know the number of requests
of a certain type?

Many prominent online problems have been studied in this framework, including
paging \cite{DKP08,BKKKM09}, the $k$-server problem \cite{BKKKM09,EFKR11,GKL2013,RR2011},
metrical task systems \cite{EFKR11}, and the online knapsack problem \cite{BKKR2014}.  Negative results on the advice
complexity can be transferred by a special kind of reduction \cite{BHKKSS2014,BFKM2015,EFKR11}.
Moreover, advice complexity has a close and non-trivial
relation to randomization \cite{BKKK11,KK11,DBLP:journals/corr/Mikkelsen15}. 
We now define online algorithms with advice formally.

\begin{definition}[Advice Complexity]
  Let $x_1, \ldots, x_n$ be the input for an online problem $\Pi$.  An \emph{online
  algorithm with advice}, \alg, for $\Pi$ computes the output sequence $y_1, \ldots, y_n$,
  where $y_i$ is allowed to depend on $x_1, \ldots, x_{i-1}$ as well as on an
  \emph{advice string} $\phi$.  The advice, $\phi$, is written in binary on an infinite
  tape and is allowed to depend on the request sequence $x_1, \ldots, x_n$.  The
  advice complexity of \alg is the largest number of advice bits it reads from
  $\phi$ over all inputs of length at most~$n$.
\end{definition}

Our paper is devoted to both creating online algorithms with advice for \OS that
achieve a certain output quality while using a certain number of advice bits,
and to show that such algorithms cannot exist if the advice complexity is below
some certain threshold.

Most of the work in advice complexity theory considers problems where at least $n$
advice bits are required for an algorithm to be optimal.  Here, we study a problem
where only $\log_2 n$ bits give an optimal algorithm.
We investigate how this problem behaves when the number of advice bits is in the interval $[1, \log_2 n]$.
For the ease of presentation, we assume that $\log_2 n$ is integer.

\section{Related Work}

The search problem in an offline setting, \ie, where the set of prices is known
in advance, can easily be solved optimally in time $O(n)$.  However, for a lot
of online environments such as stock trading and foreign exchange, decisions
should be made even though there is no knowledge of the future prices of the
currencies.  These problems are intrinsically online. 

The most common approaches are Bayesian.  These approaches rely on a prior
distribution of prices where the online algorithm computes a certain
\textit{reservation price} based on the distribution.  The trader accepts any
price that is larger than or equal to the reservation price.  If this certain
price is not met, the player has to trade on the last day (according to \cref{def:online_search}).  Throughout this paper, \algo{p}
denotes the algorithm that accepts the first price it sees that is at least~$p$.

Since the prior distribution of prices is not necessarily known in advance,
El-Yaniv et al.\ \cite{El-Yaniv2001} proposed to measure the quality of online
trading algorithms using competitive analysis.  Moreover, for some assets, the
goal is not just to increase the profit but to minimize the loss by considering
the possible worst-case scenarios in the market.  Competitive analysis in
financial problems such as \OS can provide a guaranteed performance measure for
the trader's profit.
The best deterministic online algorithm with respect to competitive analysis
is $\algoinline{\sqrt{Mm}}$, \ie, the algorithm that accepts the first
price it sees that is at least $\sqrt{Mm}$ (or it accepts $p_n$ if no
such price is ever seen). 
This algorithm has a competitive ratio of $\sqrt{M/m}$, which is provably the
best competitive ratio any deterministic online algorithm without advice can
achieve \cite{BE98}. 

Boyar et al.\ \cite{Boyar2012} studied how the problem behaves when applying
a variety of difference performance measures (and not just competitive ratio).

\section{Advice for the Online Search Problem}

In this section, we explore the advice complexity of \OS.  We start by studying
how much advice is necessary and sufficient in order to obtain an optimal
output.  After that, we study general $c$-competitiveness.

\subsection{Advice for Optimality}

It is possible for an algorithm to be optimal using $\log_2 n$ bits of advice
if $n$ is known in advance by simply encoding the day where the largest price is offered.
If $n$ is not known in advance, it has to be encoded with a self-delimiting encoding,
for example, by writing the length of $\log_2 n$ in unary followed by $\log_2 n$.  This requires $2 \log_2 n$ bits \cite{BKKKM09}.

Moreover, optimality can also be achieved by encoding the value of $p_{\text{max}}$  using $O(\log_2(M/m))$
bits, but since $M$ and $m$ can be arbitrarily large, this may be very expensive.
We now give a complementing lower bound.

\begin{theorem}\label{thm:os_lb_opt1}
  At least $\log_2 n$ bits of advice are necessary to obtain an optimal
  solution for \OS. This holds even if $n$ is known to the algorithm.
\end{theorem}

\begin{proof}
  We use that an algorithm with $b$ advice bits can be viewed as dealing with the best of $2^b$ algorithms without advice, for the particular instance chosen.
  First, we generate a set of request sequences $\mathcal{S}$. Then, we
  show that, for $\mathcal{S}$, there is no set of $n-1$ or fewer deterministic
  algorithms, which can ensure that at least one algorithm always gets the optimal solution.
  
  We construct the set $\mathcal{S}$ in such a way that
  each request has a unique optimal solution.  The construction is as follows.
  Let $\mathcal{S} = \{\sigma_1, \sigma_2, \ldots, \sigma_n \}$, such that
  \[ \sigma_i =(\underbrace{ m+ \delta ,m+2 \delta, \ldots, m+i \delta}_{i}, \underbrace{ m, \ldots, m}_{n-i})\;, \]
  where $\delta = (M-m)/n$.  Each $\sigma_i$ is thus a sequence of $n$
  prices that follow an increasing order until day $i$.
  Then the price drops to the minimum $m$ for the remaining $n-i$ days.  The optimal solution for
  each $\sigma_i$ clearly is to trade on day $i$ and obtain a profit of
  $m+i\delta$. 
  From the construction, it is impossible for any deterministic online algorithm to
  distinguish the request sequence $\sigma_i$ from any other sequence of requests $\sigma_j$, for
  $j>i$, until the price for day $i+1$ is offered.  This is due to the fact
  that the set of requests $\{\sigma_i, \sigma_{i+1}, \ldots, \sigma_n\}$ have
  the same prices offered from day $1$ up to day $i$.  
  Since we have $n$ such input instances with different optimal solutions, 
  and fewer than $n$ algorithms, there is one algorithm that gets chosen for
  at least one the above instances.  Clearly, this algorithm cannot be optimal
  for both these instances.  Thus, any online algorithm with advice
  needs $\log_2 n$ bits of advice to identify the actual
  input from these $n$ possible cases.
  \qed
\end{proof}

Note that if it is required that the prices are integral, this construction still works by picking $m$ and $M$ such that
$\delta$ is an integer.
 
\subsection{Advice for $c$-Competitiveness}

Next, we investigate the advice complexity of \OS if we have less than
$\log_2 n$ advice bits.  This means we study a tradeoff between the number $b$
of advice bits supplied and the competitive ratio $c$ obtainable.
Recall that, without advice bits, the optimal trader
strategy is to use \algo{p}, where the reservation price is $p = \sqrt{Mm}$.

Before we present the upper bounds for online algorithms with advice for \OS that achieve $c$-competitiveness, we give a
simple intuition behind our strategy.
We can think of it as having $2^b$ deterministic algorithms with
different reservation prices.  The computation of each reservation price $p_i$
is obtained by computing the solution of the following equation.

\[ \frac{p_1}{m} = \frac{p_2}{p_1} = \ldots =\frac{p_{2^i}}{p_{2^i-1}} =\ldots = \frac{M}{p_{2^b}} \]

\begin{theorem} \label{upper}
  For every $b>0$, there exists an online algorithm with advice for \OS which
  reads $b$ bits of advice and achieves a competitive ratio of at most
  $(M/m)^{\frac{1}{2^b +1}}$.  This holds even if $n$ is unknown.
\end{theorem}

\begin{proof}
  We describe an algorithm \alg with advice which reads $b$ bits of advice and achieves the claimed competitive ratio.
  First, the oracle simulates the algorithms 
  \[ \algo{m^{\frac{2^b+1-i}{2^b+1}}M^\frac{i}{2^b+1}} \]
  for $i=1, \ldots, 2^b$.
  Let $A$ denote the set of these algorithms.
  Then, it writes the value of $i$ for the algorithm that achieves the best competitive ratio.
  We argue that at least one of the algorithms gets a competitive ratio of at most
  \[ \left( \frac{M}{m} \right)^{\frac{1}{2^b+1}}\;. \]
  
  We have three cases for $p_{\text{max}}$. The first case is when $p_{\text{max}}<m^{\frac{2^b}{2^b+1}}M^\frac{1}{2^b+1}$.
  Here, each algorithm in $A$ will get the price offered on the last day, which is at least $m$.
  The competitive ratio for \alg is at most
  \[ \frac{m^{\frac{2^b}{2^b+1}}M^{\frac{1}{2^b+1}}}{m}=\left( \frac{M}{m} \right)^{\frac{1}{2^b+1}}\;. \]
  
  The second case is when $p_{\text{max}} \geq m^{\frac{1}{2^b+1}}M^\frac{2^b}{2^b+1}$.
  In this case,
  \[ \algo{m^{\frac{1}{2^b+1}}M^\frac{2^b}{2^b+1}} \]
  gets a price of at least $m^{\frac{1}{2^b+1}}M^{\frac{2^b}{2^b+1}}$.
  Since \opt gets at most $M$, the competitive ratio for \alg is again at most
  \[ \frac{M}{m^{\frac{1}{2^b+1}}M^{\frac{2^b}{2^b+1}}}=\left( \frac{M}{m} \right)^{\frac{1}{2^b+1}}\;. \]
  
  The last case is when $m^{\frac{2^b+1-i}{2^b+1}}M^\frac{i}{2^b+1} \leq p_{\text{max}} < m^{\frac{2^b-i}{2^b+1}}M^\frac{i+1}{2^b+1}$ for some $i<2^b$.
  In this case, 
  \[ \algo{m^{\frac{2^b+1-i}{2^b+1}}M^\frac{i}{2^b+1}} \]
  gets at least its reservation price. 
  Thus, also here, the competitive ratio for \alg is at most
  \[ \frac{m^{\frac{2^b-i}{2^b+1}}M^\frac{i+1}{2^b+1}}{m^{\frac{2^b+1-i}{2^b+1}}M^\frac{i}{2^b+1}}=\left( \frac{M}{m} \right)^{\frac{1}{2^b+1}}\;. \]
  
  All in all, we have shown that, in each case, \alg obtains a competitive ratio of at most $(M/m)^{\frac{1}{2^b+1}}$ as we claimed.
  \qed
\end{proof}

We now present a matching lower bound.

\begin{theorem} \label{lower}
  Let \alg be an algorithm with advice for \OS which reads $b<\log_2 n$ bits of advice.
  The competitive ratio of \alg is at least $(M/m)^{\frac{1}{2^b+1}}$.
\end{theorem}

\begin{proof}
  For any given $b < \log_2 n$, let $\alg$ be an algorithm with advice that reads at most $b$ bits of advice.  Again, we view this advice as $2^b$ deterministic online algorithms.
  We now give a class of request sequences that ensure that each of them gets a competitive ratio of at least
  \[\left( \frac{M}{m} \right)^{\frac{1}{2^b+1}}\;.\]
  
  Consider the sequence $( p_1, p_2, \ldots , p_{2^b} )$ with
  \[ p_i=m^{\frac{2^b+1-i}{2^b+1}}M^{\frac{i}{2^b+1}}\;. \]
  The adversary simulates all $2^b$ algorithms on this sequence.
  We consider two cases.
  If a request $p_i$ is rejected by all algorithms, it requests $p_1,p_2,\dots,p_i$ followed by
  requests that are all equal to $m$.
  For the first case, assume that there exists a request $p_i$, which is rejected by all $2^b$ algorithms.
  The remaining requests are all  $m$. This means that \alg gets a price of at most $p_{i-1}$ (the largest request that was not $p_{i}$) while \opt gets a price of $p_i$.
  Note that, if the first request is rejected, \alg gets a price of at most $m=p_0$.
  In this case, the competitive ratio for \alg is at least
  \[ \frac{p_i}{p_{i-1}}=\frac{m^{\frac{2^b+1-i}{2^b+1}}M^{\frac{i}{2^b+1}}}{m^{\frac{2^b+2-i}{2^b+1}}M^{\frac{i-1}{2^b+1}}}=\left(\frac{M}{m}\right)^{\frac{1}{2^b+1}}\;. \]
  
  Thus, \alg cannot obtain a competitive ratio which is better than $(M/m)^{\frac{1}{2^b+1}}$ if a request is rejected by all the algorithms.
  
  Next, we consider the second case.  Here, every request in $\sigma$ is accepted by some algorithm. Since there are $2^b$ requests in $\sigma$, it follows that all algorithms accept a price that 
  is at most $p_{2^b}$.  Since $2^b < n$, the adversary can still make a request.  The final request is then $M$.  The competitive ratio for \alg is therefore bounded from below by
  \[ \frac{M}{p_{2^b}}=\frac{M}{m^{\frac{1}{2^b+1}}M^{\frac{2^b}{2^b+1}}}=\left(\frac{M}{m}\right)^{\frac{1}{2^b+1}}\;. \]
  
  In both cases, \alg has a competitive ratio of at least $(M/m)^{\frac{1}{2^b+1}}$ as claimed by the theorem.
  \qed
\end{proof}

\section{Advice and Randomization}

Randomization is often used to improve the competitive ratio of online
algorithms (in expectation).   Here, the online player is allowed to base some
of its answers on a random source.  An oblivious adversary knows the algorithm,
but not the outcome of the random decisions.  To provide an improvement over the lower bound of deterministic
online algorithms for \OS, El-Yaniv et al.\ \cite{El-Yaniv2001} provided an
upper bound by presenting a randomized algorithm with an expected competitive ratio of
$\log_2(M/m)$.  Lorenz
et al.\ \cite{Lorenz2009} provided an asymptotically matching lower bound of
$(\log_2(M/m))/2$ for randomized online algorithms for \OS.

In this section, we compare the power of advice to the ability of an
online algorithm to access random bits for \OS. The competitive ratio of
online algorithms with advice (with an increasing number of advice bits) is shown in
\cref{fig:plot_advice_os}. 
We fixed a fluctuation ratio $M/m$, and we
highlighted the competitive ratio of the best deterministic algorithm, \ie, $(M/m)^{\frac{1}{2}}$,
and the corresponding upper (\ie, $\log_2(M/m)$) and lower (\ie, $\log_2(M/m)/2$) bounds of randomized algorithms for \OS.

\begin{figure}[t]
\begin{centering}
  \begin{tikzpicture}[scale=1.1]
    \begin{axis}[ 
	xlabel={{\small advice bits $b$}\rule{0cm}{4mm}},
	ylabel={{\small competitive ratio $c$}\rule[-5mm]{0cm}{1cm}},
	ylabel style={anchor=south},	
	ymin=0, ymax=11, 
	grid=both,
	grid style={line width=.1pt, draw=gray!10},
	major grid style={line width=.2pt,draw=gray!50},
	minor tick num=5,
	xticklabels={0,   $b^*$, $\log_2\left( \frac{M}{m} \right) $}, xtick={0, 1.503, 6.64},
	yticklabels={0, 1 ,$\frac{\log_2\left( \frac{M}{m} \right)}{2}$ , $\log_2\left( \frac{M}{m} \right)$, $\left(\frac{M}{m}\right)^\frac{1}{2}$}, ytick={0, 1,  3.32, 6.64, 10}
    ]   
    \addplot [dotted,thick, domain=0:10]{(100)^(1/2)};
    \addplot [dashed,thick, domain=0:10]{(100)^(1/(2^x+1))};
    \addplot [dotted,thick, domain=0:10]{ln(100)/ln(2)};
    \addplot [dotted,thick, domain=0:10]{((ln(100))/ln(2))/2};
    \addplot [dotted,thick, domain=0:10]{1};  
    \legend{, $\left( \frac{M}{m} \right)^{\frac{1}{2^b+1}}$}
  \end{axis}
  \end{tikzpicture}
  \caption{Plot comparing the competitive ratio of the online algorithm with advice with respect to the lower bound for deterministic and randomized algorithms.}
  \label{fig:plot_advice_os}
  \end{centering}
\end{figure}
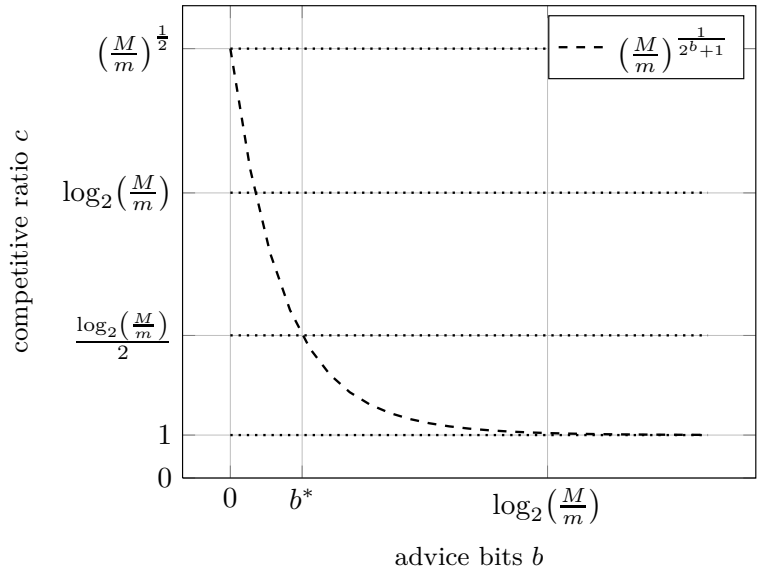

It is interesting to point out that, with the number of advice bits greater than 
\[ b^* = \log_2\left( \frac{\log_2 ( M/m )}{\log_2\left( \frac{ \log_2 ( M/m )}{2}\right)}-1\right)\;, \]
our online algorithm for \OS outperforms the lower bound of randomized online algorithms.
And as we increase the number of advice bits, the better the competitive ratio we get.
In the plot shown in \cref{fig:plot_advice_os}, we considered a fluctuation ratio $M/m = n$. 
Note that the competitive ratio is asymptotic to $1$, but 
it is actually possible to get an optimal solution with $\log_2 n$
advice bits.

\section{Conclusion and Future Work}

We studied the advice complexity of \OS and determined upper and lower
bounds on the advice complexity to achieve both optimality and $c$-competitiveness.  We presented
a tight lower bound  of $\log_2 n$ for the number of advice needed by any
online algorithm to obtain optimal solutions, as shown in \cref{thm:os_lb_opt1}. 
We also provided a strategy with $b$ bits of advice and achieved a tight bound
of $(M/m)^{\frac{1}{2^b+1}}$ for the competitive ratio as
shown in \cref{upper,lower}. 

We compared the power of advice and randomization in terms
of competitive ratio.  The comparison of the  competitive ratio is shown
in \cref{fig:plot_advice_os}.

For future work, it would be interesting to extend the results to the \OWT problem with
advice.  It is known that \OS and the \OWT are closely related.  In fact,
they are equivalent in the sense that, for every  randomized algorithm for \OS,
there exists an equivalent deterministic algorithm for \OWT
\cite{El-Yaniv2001}.  Although randomization significantly improved the
competitive ratio of algorithms for \OS, it can be shown that it cannot help to improve the
competitive ratio of algorithms for \OWT. It would be interesting to investigate the tradeoff
between advice and competitive ratio in \OWT.

\end{document}